\documentclass[11pt,a4paper]{article}
\usepackage{a4wide,amsfonts,amsmath,latexsym,amsthm,amssymb,euscript,eufrak,graphicx,units,mathrsfs, color,setspace}

\usepackage{float}
\newfloat{figure}{H}{lof}
\floatname{figure}{\figurename}

\usepackage[french,english]{babel}
\usepackage[T1]{fontenc}

\DeclareMathAlphabet{\eufrak}{U}{}{}{}  
\SetMathAlphabet\eufrak{normal}{U}{euf}{m}{n}
\SetMathAlphabet\eufrak{bold}{U}{euf}{b}{n}
\newcommand{\ndot}{\raisebox{.4ex}{.}}

\newtheorem{prop}{Proposition}[section]

\newtheorem{theorem}[prop]{Theorem}
\newtheorem{lemma}[prop]{Lemma}
\newtheorem{corollary}[prop]{Corollary}

\newtheorem{assumption}[prop]{Assumption}

\theoremstyle{definition}
\newtheorem{example}[prop]{Example}
\newtheorem{remark}[prop]{Remark}

\numberwithin{equation}{section}

\def\E{\mathbb{E}}
\def\P{\mathbb{P}}
\def\real{\mathbb{R}}

\def\F{\mathcal{F}}
\def\1{\textbf{1}}
\def\ind#1{\textbf{1}_{\{#1\}}}

\def\eqlaw{\overset{\mathcal{L}}{=}}

\newcommand{\eps}{\varepsilon}

\newcommand{\be}{\begin{equation}}
\newcommand{\ee}{\end{equation}}
\newcommand{\bde}{\begin{displaymath}}
\newcommand{\ede}{\end{displaymath}}
\newcommand{\beq}{\begin{eqnarray*}}
\newcommand{\eeq}{\end{eqnarray*}}
\newcommand{\beqa}{\begin{eqnarray}}
\newcommand{\eeqa}{\end{eqnarray}}
\newcommand{\bel }{\left\{\begin{array}{ll}}
\newcommand{\eel}{\cr \end{array} \right.}

\newcommand{\bex}{\begin{ex} \rm }
\newcommand{\eex}{\end{ex}}

\def\R{\mathbb R}
\def\N{\mathbb N}
\def\E{\mathbb E}
\def\F{{\cal F}}

\def\P{\mathbb P}

\def\cal#1{\mathcal{#1}}

\def\b{\textcolor{blue}}

\definecolor{ying}{rgb}{0.8, 0.0, 0.04}

\DeclareSymbolFontAlphabet{\mathrsfs}{rsfs}

\author{Caroline Hillairet\footnote{ENSAE Universit\'e Paris Saclay, CREST,
5  avenue Henry Le Chatelier
91120 Palaiseau, France. \; Email: \texttt{caroline.hillairet@ensae.fr}} \and Ying Jiao\footnote{Universit\'e Claude Bernard - Lyon 1, Institut de Science Financi\`ere et d'Assurances, 69007 Lyon France. \; Email: \texttt{ying.jiao@univ-lyon1.fr}} \and Anthony R\'eveillac\footnote{INSA de Toulouse, IMT UMR CNRS 5219, Universit\'e de Toulouse, 135 avenue de Rangueil 31077 Toulouse Cedex 4 France. \; Email: \texttt{anthony.reveillac@insa-toulouse.fr}}}

\title{Pricing formulae for derivatives in insurance using the Malliavin calculus\footnote{The authors acknowledge Projet PEPS \'egalit\'e (part of the European project INTEGER-WP4) "Approximation de Stein : approche par calcul de Malliavin et applications \`a la gestion des risques financiers" for financial support.}}

\begin{document}

\maketitle

\allowdisplaybreaks

\begin{abstract}
\noindent
In this paper we provide a valuation formula for different classes of actuarial and financial contracts which depend on a general loss process, 
by using the Malliavin calculus. In analogy with the celebrated Black-Scholes formula, we aim at expressing the expected cash flow in terms of a building block. The 
former is related to the loss process which is a cumulated sum indexed by a doubly stochastic Poisson process of claims allowed to be dependent on the intensity and the jump times of the counting process. For example, in the context of Stop-Loss contracts the building block is given by the distribution function of the terminal cumulated loss, taken at   the Value at Risk when computing the Expected Shortfall risk measure.
\end{abstract}

\section{Introduction}


Risk analysis in the context of insurance or reinsurance is often based on the study of properties of a so-called \textit{cumulative loss process} $L:=(L_t)_{t\in[0,T]}$ over a period of time $[0,T]$ where $T>0$ denotes the maturity of a contract. Usually, $L$ takes the form
$$ L_t:=\sum_{i=1}^{N_t} X_i, \quad t\in [0,T],$$
where $N:=(N_t)_{t\in[0,T]}$ is a counting process, and the random variables $(X_i)_{i\in\N^*}$ represent the amount of the claims. A typical contract in reinsurance is the \textit{Stop-Loss contract} that offers protection against an increase in either (or both) severity and frequency of a company's
loss experience.
More precisely,  Stop-loss contracts provide to its buyer (another insurance company) the protection against losses which are larger than a given level $K$ and its payoff function is given by a ``call'' function. In some cases, there is also an upper limit given by some real number $M$, which specifies the maximal reimbursement amount. Thus the payoff of such a contract is given by
\begin{equation}\label{payoff stop loss}\Phi(L_T)= \begin{cases} 0, \; &\textrm{ if } L_T< K;\\ L_T-K,  &\textrm{ if } K \leq L_T < M;\\ M {-K},  &\textrm{ if } L_T\geq M. \end{cases}\end{equation}
In full generality the risk carried out by the claims is neither hedgeable nor related to a financial market, hence the premium of the Stop-Loss is equal to $\E[\Phi(L_T)]$ which immediately re-writes as
\begin{equation}
\label{eq:intropremium}
\E[\Phi(L_T)] = \E\left[L_T \ind{L_T \in [K,M]}\right] - K \P\left[L_T \in [K,M]\right] + (M-K) \P\left[L_T \geq M\right].
\end{equation}
There is a large number of papers describing how to approximate the compound distribution function of the cumulated loss $L_T$,
 and to compute the Stop-Loss premium. The aggregate claims distribution function can in
some cases be calculated recursively, using, for
example, the Panjer recursion formula, see Panjer  \cite{panjer1981}
and  Gerber \cite{gerber1982}. Various approximations of Stop-Loss reinsurance premiums are described in the literature, some of them assuming a specific dependence structure.\\
In analogy with the celebrated Black-Scholes formula, we aim in this paper to express the first term of the right-hand side of (\ref{eq:intropremium}) in terms of a building block which represents the distribution function of the terminal loss $L_T$. This feature is hidden in the Black-Scholes model since the terminal value of the stock has an explicit lognormal distribution. More specifically, we aim in computing $\E\left[L_T \ind{L_T \in [K,M]}\right]$ by using the building block $x\mapsto\P\left[L_T \in [K-x,M-x]\right]$. Note that, on the credit derivative market, the payoff function \eqref{payoff stop loss} can also be related to Collateralized Debt Obligations (CDOs) where there are several tranches, and so several $K$ and $M$ levels, which are expressed in proportion of the underlying which is the loss of a given asset portfolio.\\\\
\noindent
Stop-Loss contracts are the paradigm of reinsurance contracts, but we aim in dealing with more general payoffs whose valuation involves the computation of the quantity
\begin{equation}
\label{eq:generalpayoff}
\E\left[\hat L_T h\left(L_T\right)\right],
\end{equation}
where $h:\real_+\to \real_+$ is a Borelian map and where $\hat L$ is of the form $ \hat L_T := \sum_{i=1}^{N_T} \hat X_i, $
involving claims $\hat X_i$ which are related to the ones $X_i$ of the original loss $L_T$. To be more precise, $\hat L_T$ will be the effective loss covered by the reinsurance company whereas $L_T$ is the loss quantity that activates the contract. Typical examples will be given in Section \b{\ref{subsec:lossprocess}}. Once again, this is similar to the valuation of CDOs tranches where the recovery rate is often supposed to be a random variable of beta distribution with mean $40\%$ 
whereas the realized rate, often revealed only after the formal bankruptcy, does not necessarily match with this value
.\\\\
\noindent
In this paper we provide an exact formula for (\ref{eq:generalpayoff}) in terms of the building block $x\mapsto \E\left[h(L_T+x)\right]$ (or of a related quantity for the more general situation (\ref{eq:generalpayoff}), see (\ref{eq:definitionvarphi}) for a precise statement). This goal will be achieved by using one of the Malliavin calculus available for jump processes. Before turning to the exposition of the model, we emphasize that this methodology goes beyond the analysis of pricing and finds for instance application in the computation of the Expected Shortfall of contingent claims in the realm of risk measures. Indeed, the expected shortfall is a useful risk measure, that takes into account the size of the expected loss above the value at risk. Formally it is defined as
$$ES_{\alpha}(-L_T) = \E\left[ -L_T\middle\vert -L_T > V@R_{\alpha}(-L_T) \right], \quad \alpha \in (0,1).$$
As it is well-known, the expected shortfall coincides with Average Value at Risk (AV@R), that is
$$ ES_{\alpha}(-L_T) =  AV@R(-L_T):=\frac{1}{1-\alpha} \int_{\alpha}^1 V@R_s(-L_T) ds, $$
if and only if $\P[-L_T\leq q_{-L_T}^+(t)]=t$, $t\in (0,1)$, where $q_{-L_T}^+(t)$ denotes the quantile of level $t$ of $-L_T$ (see Section \ref{section:ES} for a precise definition). However, already in the trivial example where the size claims $X_i$ are constant equal to $1$, this property fails as $L_T=N_T$ is a Poisson random variable which exhibits a discontinuous distribution function. However, our approach gives an alternative explicit computation of $\mathbb E[L_T\ind{L_T<\beta}]$ and thus of $ES_{\alpha}(-L_T)$ as
\[ES_{\alpha}(-L_T) =\frac{-\mathbb E[L_T\ind{L_T<\beta}]}{\mathbb P(L_T<\beta)}, \quad \beta:=-V@R_{\alpha}(-L_T). \]

\noindent
We conclude this section with some comments about the modeling of the claims $X_i$ and $\hat X_i$. In the classic Cramer-Lundberg model, the claims are independent and identically distributed (i.i.d.) and in addition independent of the counting process $N$ which happens to be an inhomogeneous Poisson process. In this work we consider a doubly stochastic Poisson process $N$ and we allow dependency between the size of the claims, their arrivals and the intensity of $N$. In particular we do not assume a Markovian setting.
The impact of  certain dependence structure on  the Stop-Loss premium is studied in the reinsurance literature, such as in Albers  \cite{albers1999}, Denuit  et al. \cite{denuit2001} or   De Lourdes Centeno \cite{de2005dependent}, but those works usually assume dependency between the  successive claim sizes and
 the arrival intervals.
Nevertheless, in the ruin theory literature, some contributions already propose explicit dependencies among inter-arrival times  and  the claim sizes, such as  Albrecher and Boxma  \cite{albrecher2004ruin}, Boudreault, Cossette, Landriault and Marceau \cite{boudreault2006risk} and related works. A general framework of dependencies is proposed  by  Albrecher,  Constantinescu, and Loisel \cite{albrecher2011explicit} in which  the dependence arises via mixing through a so-called frailty parameter. Recently, Albrecher et al.  \cite{albrecher2017queue} extend duality results that
relate survival and ruin probabilities in the insurance risk model to waiting time
distributions in the 'corresponding' queueing model.
 The risk processes    have a counterpart in workload models of queueing theory,
and a similar mixing dependencies  structure is considered  in a queueing context. Besides, our framework extends the mixing approach of \cite{albrecher2011explicit} and  \cite{albrecher2017queue} by allowing non-exchangeable family of random variables for the claims amounts. In a similar way, in the credit risk modeling we can also suppose that the recovery rate depend on the underlying default intensity such as in Bakshi, Madan and Zhang \cite{BMZ2006}.\\\\
\noindent
We proceed as follows. We first make clear in Section \ref{section:model} our model for the loss process and present the insurance contracts for which we will propose a pricing formula. The latter will be stated and proved as Theorem \ref{th:main} in Section \ref{section:main}. Particular cases of this result to several types of contracts in insurance are also given in this section. Finally, explicit examples are presented in Section \ref{section:examples}.

\section{Model Setup}
\label{section:model}
In this section, we describe the loss process and the associated reinsurance contracts we will study. Throughout this paper, $T$ will denote a positive finite real number which represents the final horizon time.
\subsection{The Loss process}\label{subsec:lossprocess}
We begin by introducing the loss process $L:=(L_t)_{t\in[0,T]}$ where the size of claims and their arrival times are correlated.  Let  $(N_t)_{t\in[0,T]}$ be  a Cox process (also called doubly stochastic Poisson process)  with random intensity $(\lambda_t)_{t\in[0, T]}$,  whose jump times, denoted by $(\tau_i)_{i\in \N^*}$, model the arrival times of the claims. We suppose that the claim size $X_i$ depends on both  the cumulated intensity  defined by $\Lambda_t:=\int_0^t \lambda_s ds$ and the claim arrival time $\tau_i$. Moreover, it will also depend on some random variable $\eps_i$ where we suppose that $(\eps_i)_{i\in\N^*}$ is a sequence of positive i.i.d. random variables  independent of the Cox process $N$. More precisely, the loss is given by
 \begin{equation}
\label{eq:L}
L_t := \sum_{i=1}^{N_t} X_i \, e^{-\kappa (t-\tau_i)}, \quad \mbox{ with } X_i:=f(\tau_i, \Lambda_{\tau_i},\eps_i), \quad t\in [0,T],
\end{equation}
where  $\kappa$ is the  discounting factor and $f:\real_+^3 \to \real_+$ is a bounded  deterministic function. We provide several examples as below.
\begin{example}
\begin{enumerate}
\item In the classic ruin theory, the claim size is often supposed to be independent of the arrival and the intensity process. In this case, we have   $f(t, \ell, x)= x$.

\item In the second example, we suppose that the dependence of $f$ on the exogenous factor $\eps$ is linear and the linear coefficient is a function of the cumulated intensity $\Lambda$ rescaled by time, \textit{i.e.}, $\frac{\Lambda_t}{t}$, which stands for some mean level of the intensity. For instance,  let $$f(t,\ell,x)= \sqrt{\frac{\ell}{t}} x. $$
In this example, if $\varepsilon_i$ follows an exponential distribution with parameter $1$, then  $X_i=f(\tau_i,\Lambda_{\tau_i}, \varepsilon_i)$ follows an exponential distribution with parameter $ \sqrt{\frac{\tau_i}{\Lambda_{\tau_i}}}$ conditionally to the vector $(\tau_i,\Lambda_{\tau_i})$.
\end{enumerate}
\end{example}

\subsubsection{Generalized loss process}

We can also consider a more general case where the realized  claim sizes $(X_i)_{i\in\N^*}$ are not exactly the ones that are  computed to activate the reinsurance contract.  More precisely, assume that in
addition to the  factors $(\varepsilon_i)_{i\in\N^*}$, there exists  a family of i.i.d. positive random variables $(\vartheta_i)_{i\in\N^*}$ which may depend on the random variables $\eps_i$'s.
Let $g:\real_+^4\to \real_+$ be a deterministic bounded function.
We can define a modified cumulative loss process as
\begin{equation}
\label{eq:lossbis}
\hat L_t := \sum_{i=1}^{N_t} g(\tau_i, \Lambda_{\tau_i},\eps_i,\vartheta_i) e^{-\kappa(t-\tau_i)}, \quad t\in [0,T].
\end{equation}
\noindent
More precisely, although the insurance contract is triggered by the loss process $L$, the compensation amount can depend on some other exogenous factors $(\vartheta_i)_{i\in\N^*}$. This would mean for instance that the amounts  $\vartheta_i$'s are much lower than the $\eps_i$'s. A typical example is given by the housing insurance market on the American East Coast. Indeed, this region is seasonally exposed to hurricanes of different magnitudes. Most of the damages impacts the houses of the insured who may as well buy contracts on other belongings such as cars which are much less valuable. After a hurricane episode, the re-insurance Stop-Loss contract will be activated on the basis of the total damages $L_T$ on the houses (which are represented by the claims $\eps_i$) whereas the effective damages $\hat L_T$ will also include all other insured belongings (which would be modeled by the $\vartheta_i$). In the special case where the function $g$ does not depend on the fourth variable, the general loss $\hat L_t$ reduces to the standard loss defined in \eqref{eq:L}. We give below some examples of the joint distribution $(\eps_i,\vartheta_i)$.

\begin{example}
\begin{enumerate}
\item The first natural case is that $\varepsilon_i$ and $\vartheta_i$ are independent random variables. For example, each of them can follow an exponential distribution (or Erlang distribution) with different positive parameters $\theta_1$ and $\theta_2$.

\item We can introduce dependence between $\varepsilon_i$ and $\vartheta_i$ by using the mixing method in \cite{albrecher2011explicit}. Let  $\varepsilon_i$ and $\vartheta_i$ follow Pareto marginal distributions respectively and a dependence structure  according to a Clayton copula (according to Example 2.3 in  \cite{albrecher2011explicit}, this can be achieved by mixing the two Pareto marginal distributions where the mixing parameter follows a Gamma distribution).

\item Case of explicit dependence : let $\varepsilon_i$ follow a Pareto distribution and $\vartheta_i$ follow a Weibull distribution with form or scaling parameter depending of  $\varepsilon_i$.

\end{enumerate}
\end{example}

\subsection{Reinsurance contracts and related quantities}
\subsubsection{{ Generalized Stop-loss Contrats}}\label{section:stoploss}
We have seen in the introduction the Stop-Loss contract whose payoff is given by $\Phi(L_T)$ where $\Phi$ has been defined in (\ref{payoff stop loss}) and corresponds to a call spread, that is, the difference of two call functions.
Our approach allows us to go beyond the case of the Stop-Loss contract. Consider now a contract where the reinsurance company pays
\begin{equation}\label{payoff generalized stop loss} \widetilde\Phi(L_T,\hat L_T)=\begin{cases} 0, &\textrm{ if } L_T\leq K\\
\hat L_T-K, &\textrm{ if } K \leq L_T \leq M\\
M -K, &\textrm{ if } L_T\geq M \end{cases},\end{equation}
with $\hat L_T$ defined in \eqref{eq:lossbis} if the \textit{a priori} loss $L_T$ excesses some amount $K$ or belongs to some interval $[K,M]$.
Then the price of such a contract is :
\begin{equation}\label{decomposed general stop loss} \E\left[\hat L_T \ind{L_T>K}\right] - K \P\left[L_T \in [K,M]\right] + (M-K) \P\left[L_T \geq M\right]. \end{equation}

%
%

\subsubsection{Expected Shortfall}
\label{section:ES}
The expected shortfall is a useful risk measure which takes into account the size of the expected loss above the value at risk.
We recall the Expected Shortfall with level $\alpha$ as
  $$ES_{\alpha}(-L_T) = \E\left[ -L_T\middle\vert -L_T > V@R_{\alpha}(-L_T) \right], \quad \alpha \in (0,1).$$
where the definition of $V@R$ is\[V@R_{\alpha}(X)=-q_X^+(\alpha)=q_{-X}^-(1-\alpha)\]with
\[q_X^+(t)=\inf\{x|\,\mathbb P[X\leq x]>t\}=\sup\{x|\,\mathbb P[X<x]\leq t\}\]
\[q_X^-(t)=\sup\{x|\,\mathbb P[X<x]<t\}=\inf\{x|\,\mathbb P[X\leq x]\geq t\}.\]
\noindent
It is well known that $ES_{\alpha}(X)$ is equal to $AV@R(X):=\frac{1}{1-\alpha} \int_{\alpha}^1 V@R_s(X) ds  $ if and only if $\P[X\leq q_{X}^+(t)]=t$, $t\in(0,1)$, which is in particular satisfied if the distribution function of $X$ is continuous (see \textit{e.g.} \cite[Relation (4.38)]{Follmer}).  However, the latter property fails  already in the case where the size claims $X_i$ are constant. Thus one can not rely on the above relation and has to compute directly the conditional expectation $ES_{\alpha}(-L_T)$.

We will provide an alternative expression for the expected shortfall. We denote by $\beta:=-V@R_{\alpha}(-L_T)$, then
\[ES_{\alpha}(-L_T) =\frac{-\mathbb E\left[L_T\ind{L_T<\beta}\right]}{\mathbb P[L_T<\beta]}\]
where
\[\beta=q^+_{-L_T}(\alpha)=\inf\{x|\,\mathbb P[L_T>-x]>\alpha.\}\]
So once again the key term to compute turns out to be the expectation $\mathbb E\left[L_T\ind{L_T<\beta}\right]$.

\subsection{General payoffs}

More generally, we are interested in computing quantities of the form
$$ \E\left[\hat L_T h\left(L_T\right)\right],$$
where $h:\real_+\to \real_+$ is a Borelian map with $\E[h(L_T)]<\infty$. Since in  our model, the counting process is given by a Cox process with stochastic intensity, the building block becomes the following mapping by using the conditional expectation
$$x\mapsto \E\left[h(L_T+x)\vert (\lambda_t)_{t\in[0,T]}\right].$$ Note that the examples of Section \ref{section:stoploss} (respectively of Section \ref{section:ES}) are contained in this setting by choosing $h:=\textbf{1}_{[K,M]}$ for some $-\infty\leq K<M\leq +\infty$ (respectively $h:=\textbf{1}_{[-\infty,\beta]}$ and $\hat L_T=L_T$).\\\\
\noindent
Our approach calls for some stochastic analysis material that we present in the next section.

\section{The pricing formulae using the Malliavin calculus}
\label{section:main}
In this section, we establish our main pricing formulae by using the Malliavin calculus.  To this end, we first make precise the Poisson space associated to the loss process. Then we provide basic tools for the Malliavin calculus.

\subsection{Construction of the Poisson space}

\subsubsection{The counting process and intensity process}

We recall that the loss process involves the Cox process $(N_t)_{t\in[0, T]}$ with its intensity and jump times, and the family of random variables $(\eps_i)_{i\in\N^*}$.  We begin by introducing a general counting process which will be useful for the construction of $(N_t)_{t\in[0, T]}$ on a suitable space. Let $\Omega_1$ be the set of  (finite or infinite) strictly increasing sequences in $]0,+\infty[$\,. We define a continuous-time stochastic process $\mathcal C$ on the set $\Omega_1$ as
\[\forall\,(t,\omega_1)\in [0,+\infty[\times\Omega_1,\quad \mathcal C_t(\omega_1):=\mathrm{card}([0,t]\cap\omega_1).\]Let $\mathbb F^{\mathcal C}=(\mathcal F_t^{\mathcal C})$ be the filtration generated by the process $\mathcal C$, namely $\mathcal F_t^{\mathcal C}:=\sigma(\mathcal C_s,\,s\leq t)$. It is known that there exists a unique probability measure $\mathbb P_1$ on $(\Omega_1,\mathcal F_\infty^{\mathcal C})$ under which the process $\mathcal C$ is a Poisson process of intensity $1$, that is, for every $(s,t) \in [0,+\infty)^2$, with $s<t$, the random variable $\mathcal{C}_t-\mathcal{C}_s$ is independent of $\F_s^{\mathcal{C}}$ and Poisson distributed with parameter $t-s$.\\\\
\noindent
We then consider a probability space $(\Omega_2,\mathcal{A},\P_2)$  on which is defined :
\begin{itemize}
\item[(i)] a positive stochastic process $(\lambda_t)_{t\in [0,T]}$ such that
$$ \int_0^T \lambda_s ds<+\infty, \;\;\; \P_2\text{ - a.s.}. $$
\item[(ii)] a collection of i.i.d. $\R_+^2$-valued bounded random variables $(\eps_i,\vartheta_i)_{i\in\mathbb N^*}$ and a $\R_+^2$-random variable $(\overline \eps,\overline \vartheta)$ independent from $(\eps_i,\vartheta_i)_{i\in\mathbb N^*}$, with $(\overline{\eps},\overline \vartheta) \eqlaw (\eps_1,\vartheta_1)$ (where $\eqlaw$ stands for the equality of probability distributions). We set $\mu$ the law of the pair $(\overline{\eps},\overline \vartheta)$.
\end{itemize}

\begin{assumption}
\label{ass:1}
We assume that $\lambda$ is independent of $(\eps_i,\vartheta_i)_{i\in\mathbb N^*}$, and of $(\overline{\eps},\overline{\vartheta})$.
\end{assumption}
 $\mathbb F^\lambda=(\mathcal F_t^\lambda)_{t\in[0,T]}$ be the right-continuous complete filtration generated by the stochastic process $\lambda$. Moreover, we set
\begin{equation}
\label{eq:Lambda}
\Lambda_t:=\int_0^t \lambda_s ds, \quad t\in [0,T].
\end{equation}
Let $\mathcal F^{\varepsilon,\vartheta}$ be the $\sigma$-algebra generated by $(\varepsilon_i)_{i\in\mathbb N^*}$ and $(\vartheta_i)_{i\in\mathbb N^*}$. Note that only $(\varepsilon_i)_{i\in\mathbb N^*}$ and $(\vartheta_i)_{i\in\mathbb N^*}$ will be involved in the loss process and $\overline{\varepsilon}$ and $\overline{\vartheta}$ are just independent copies which play an auxiliary role. We denote by $\mu$ the probability law of the couple $(\varepsilon_i,\vartheta_i)$.

\begin{assumption}
\label{ass:2}
Throughout this paper, we assume that : $ \Lambda_T <+\infty, \; \P_2 - a.s.. $
\end{assumption}

\subsubsection{The doubly stochastic Poisson process}

We now consider the product space $(\Omega:=\Omega_1\times\Omega_2,\F:=\F^{\mathcal{C}}_{\infty}\otimes\mathcal{A},\P:=\P_1\otimes\P_2)$. By abuse of notation, any random variable $Y$ on $\Omega_1$ can be considered as a random variable on $\Omega$ which sends $\omega=(\omega_1,\omega_2)$ to $Y(\omega_1)$. Similarly, any random variable $Z$ on $\Omega_2$ can be considered as a random variable on $\Omega$ which sends $\omega=(\omega_1,\omega_2)$ to $Z(\omega_2)$.\\\\
\noindent
We define a counting process $N:=(N_t)_{t \in[0,T]}$ on $\Omega$ by using a time change as
$$N_t(\omega_1,\omega_2):=\mathcal{C}_{\Lambda_t(\omega_2)}(\omega_1) = \mathcal{C}_{\int_0^t \lambda_s(\omega_2) ds}(\omega_1), \quad t\in [0,T], \; (\omega_1,\omega_2)\in \Omega.$$
%
Note that for any $t$, $N_t$ is $\F_\infty^{\mathcal C}\otimes\mathcal F^\lambda_T$-measurable random variable.
Moreover, for any fixed  $\omega_2$ in $\Omega_2$, $N_t(\ndot,\omega_2)$ is an inhomogeneous Poisson process on $\Omega_1$ with intensity $t\mapsto \lambda_t(\omega_2)$ with respect to the filtration $(\F^{\mathcal C}_{\Lambda_t(\omega_2)})_{t\in[0,T]}$
which reads as\footnote{By a slight abuse of notation, $ \E\left[\cdot\middle\vert \F_T^\lambda\right] := \E\left[\cdot\middle\vert \F_0^{\mathcal C}\otimes\F_T^\lambda\right]$ and $ \E\left[\cdot\middle\vert \F_T^\lambda \vee \F^{\varepsilon,\vartheta}\right] := \E\left[\cdot\middle\vert \F_0^{\mathcal C}\otimes(\F_T^\lambda\vee \F^{\varepsilon,\vartheta})\right]$.}
$$ \E\left[e^{iu (N_t-N_s)} \middle\vert \F^\lambda_s \right] = \E\left[\exp\left((e^{iu}-1)\int_s^t \lambda_r dr \right)\middle\vert \F^\lambda_s \right], \quad 0\leq s<t\leq T,$$
where $\E$ denotes the expectation with respect to the measure $\P$. For a process $(u_t)_{t\in[0,T]}$ such that :
\begin{equation}
\label{eq:processbien}
\left\lbrace
\begin{array}{l}
u_t \textrm{ is } \F\textrm{-measurable}, \quad t \in[0,T],\\
\textrm{for a.e. } \omega_2 \in \Omega_2, (u_t(\cdot,\omega_2))_{t\in [0,T]} \textrm{ is } (\F_{\Lambda_t(\omega_2)}^{\mathcal C})_{t\in [0,T]}\textrm{-predictable},\\
 \E\left[\int_0^T |u_t| dt\right]<+\infty,
\end{array}
\right.
\end{equation}
we denote by $\left(\int_0^T u_s dN_s\right)(\omega_1,\omega_2)$ the Lebesgue-Stieltjes integral of $u(\omega_1,\omega_2)$ against the measure $N(\omega_1,\omega_2)$.\\\\
\noindent
For any $i\in\mathbb N$, we let $\tau_i$ be the $i$-th jump time of the process $N$, namely
\[\forall\,\omega=(\omega_1,\omega_2)\in\Omega,\quad \tau_i (\omega):= \inf\{t>0, \; N_t=\mathcal{C}_{\Lambda_t(\omega_2)}(\omega_1) \geq i \},\]
with the convention $\tau_0=0$.

\subsection{The Malliavin integration by parts formula}
We can now state the Malliavin integration by parts formula on the product space.
For any $t\in[0,T]$, and $\omega_1\in\Omega_1$ which is of finite length or has a limit greater than $t$, we define $\omega_1 \cup \{t\}$ in $\Omega_1$ as the increasing sequence whose underlying set is the union of $\omega_1$ and $t$. The effect of this operator is to add a jump at time $t$ to the Poisson process $N$. Finally, for $\omega:=(\omega_1,\omega_2)\in\Omega$, and $t\in[0,T]$, we set
$$\omega \cup \{t\} :=(\omega_1 \cup \{t\},\omega_2),$$
provided that $\omega_1\cup\{t\}$ is well defined.
\def\empty{
\begin{equation}
\label{eq:addjump1}
(\omega_1 \cup \{t\})(A):= \omega_1(A\setminus \{t\}) + \textbf{1}_{A}(t), \quad A\in \mathcal{B}([0,T]),
\end{equation}
where
$$ \textbf{1}_{A}(t):=\left\lbrace \begin{array}{l} 1, \quad \textrm{if } t\in A,\\0, \quad \textrm{else.}\end{array}\right. $$
\noindent}
The following lemma is a direct extension of the one presented for example in \cite[Corollaire 5]{Picard_French_96} or \cite{Picard_PTRF96} (see also \cite{Privault_LectureNotes}).

\begin{lemma}
\label{Lemma:IPP}
Let $u:\Omega\times[0,T] \to \R$ be a stochastic process which enjoys (\ref{eq:processbien}),
and $F:\Omega \to \real$ be a bounded $\F$-measurable random variable. Then the stochastic process $(\omega,t) \mapsto F(\omega \cup \{t\})$ is well-defined $\P\otimes dt$-a.e. and
\begin{equation}
\label{eq:IPP}
\E\left[F \int_0^T u_s dN_s \middle\vert \F_T^\lambda \vee \F^{\eps,\vartheta}\right] = \E\left[\int_0^T u_t \; F(\cdot\cup \{t\}) \lambda_t dt \middle\vert \F_T^\lambda \vee \F^{\eps,\vartheta} \right].
\end{equation}
\end{lemma}

\subsection{The main result}

\noindent
In this section we present our main result concerning the computation of the quantity
$$ \E\left[\hat L_T h\left(L_T\right)\right],$$
where $h:\real_+\to \real_+$ is a Borelian map with $\E[h(L_T)]<\infty$ and where $L_T$ and $\hat L_T$ are respectively defined in (\ref{eq:L}) and (\ref{eq:lossbis}). We set
\begin{equation}
\label{eq:definitionvarphi}
\varphi_\lambda^h(x):=\E\left[h(L_T+x)\vert \F_T^\lambda\right], \quad x\in \real_+.
\end{equation}
It might be surprising at first glance to consider the conditional expectation given $\lambda$ in the building block. In fact, as the intensity $\lambda$ of $N$ is random, it can be compared to a Black-Scholes model with independent stochastic volatility. In that context the Black-Scholes formula would be written in terms of the conditional law of the terminal value of the stock given the volatility (which would simply be a lognormal distribution with variance given by the volatility). Recall that for the insurance contract presented in Section \ref{section:stoploss}, $h:=\textbf{1}_{[K,M]}$ and thus $\varphi_\lambda^h$ coincides with the conditional distribution function of $L_T$.\\\\
 \noindent
Before turning to the statement and the proof of the main result, note that
\begin{equation}
\label{eq:Lbis}
\hat L_T =\int_0^T \hat Z_s dN_s,
\end{equation}
with
\begin{equation}
\label{eq:Z}
\hat Z_s:=\sum_{i=1}^{+\infty} g(s,\Lambda_{s},\eps_i,\vartheta_i) e^{-\kappa (T-s)}\textbf{1}_{(\tau_{i-1},\tau_i]}(s), \quad s \in [0,T].
\end{equation}
Moreover on the set $\{\Delta_sN= 0\}$, one has
\begin{equation}\label{equ:Zs2}
\hat Z_s=g(s,\Lambda_s,\varepsilon_{1+N_s},\vartheta_{1+N_s})e^{-\kappa(T-s)}.
\end{equation}
\noindent
As $\Lambda$ is a continuous process, $\hat Z$ satisfies Relation (\ref{eq:processbien}), provided that $\displaystyle{\E\left[\int_0^T|\hat Z_t|dt\right]<+\infty}$.\\\\
\noindent
We start our analysis with the following lemma.

\begin{lemma}\label{lemme:temp1lemma}
Under Assumptions \ref{ass:1} and \ref{ass:2},  for any $t\in[0,T]$, it holds that
$$ \left(g(t,\Lambda_t,\eps_{1+N_t},\vartheta_{1+N_t}) e^{-\kappa(T-t)},L_{T}(\cdot \cup \{t\}),\lambda_t\right) \eqlaw \left(g(t,\Lambda_t,\overline\eps,\overline\vartheta) e^{-\kappa (T-t)}, L_T + f(t,\Lambda_{t},\overline \eps) e^{-\kappa (T-t)},\lambda_t\right). $$
\end{lemma}

\begin{proof}
We set
$$ L_t := \sum_{i=1}^{N_t} f(\tau_i,\Lambda_{\tau_i},\eps_{i}) e^{-\kappa (T-\tau_i)}, \quad L^{+}_t := \sum_{i=1}^{N_t} f(\tau_i,\Lambda_{\tau_i},\eps_{i+1}) e^{-\kappa (T-\tau_i)}, \quad t\in [0,T].$$
We first {precise the value of} $L_T(\omega \cup {\{t\}})$ for a fixed element $t\in (0,T)$ and for $\omega:=(\omega_1,\omega_2)$ in $\Omega$ such that { $t\not\in\omega_1$ and $ \omega_1\cup\{t\}$ is well defined (the set of such $\omega$ has probability $1$)}. By definition, we have that
\[L_T(\omega \cup \{t\})\\
= \sum_{i=1}^{N_T(\omega \cup \{t\})} f(\tau_i(\omega \cup \{t\}),\Lambda_{\tau_i(\omega \cup \{t\})}(\omega_2),\eps_{i}(\omega_2)) e^{-\kappa (T-\tau_i(\omega \cup \{t\}))}\]
Note that one has
\[\forall\,i\in\mathbb N,\quad \tau_i(\omega\cup\{t\})=\begin{cases}\tau_i(\omega),& \text{if }i\leq N_t(\omega),\\
t,&\text{if }i=N_t(\omega)+1\\
\tau_{i-1}(\omega),&\text{if }i>N_t(\omega)+1.
\end{cases}\]
Therefore we can write $L_T(\omega\cup\{t\})$ as the sum of three terms as follows
\begin{equation}\label{decomposition of LT}\begin{split}L_T(\omega\cup\{t\})&=\sum_{i=1}^{N_t(\omega)} f(\tau_i(\omega_1),\Lambda_{\tau_i(\omega)}(\omega_2),\eps_{i}(\omega_2)) e^{-\kappa (T-\tau_i(\omega_1))}\\
&\quad +f(t,\Lambda_t(\omega_2),\varepsilon_{1+N_t(\omega)}(\omega_2))e^{-\kappa(T-t)}\\
&\quad+\sum_{i=N_t(\omega)+2}^{N_T(\omega)+1}f(\tau_{i-1}(\omega_1),\Lambda_{\tau_{i-1}(\omega)}(\omega_2),\varepsilon_i(\omega_2))e^{-\kappa(T-\tau_{i-1}(\omega_1))}.
\end{split}\end{equation}
By definition, the first term in the sum is just $L_t(\omega)$. Moreover, by a change of index we can write the third term as
\begin{equation}\label{third term}\sum_{i=N_t(\omega)+1}^{N_T(\omega)}f(\tau_i(\omega),\Lambda_{\tau_i(\omega)}(\omega_2),\varepsilon_{i+1}(\omega_2))\mathrm{e}^{-\kappa(T-\tau_{i}(\omega))}=L^{+}_T(\omega)-L^{+}_t(\omega).\end{equation}
Therefore by \eqref{decomposition of LT} the following equality holds almost surely
\begin{equation}\label{equ:star}f(t,\Lambda_t,\eps_{1+N_t}) e^{-\kappa(T-t)}=(L_T(\ndot\cup\{t\})-L_t)-(L^{+}_T-L^{+}_t).\end{equation}
Moreover, from the decomposition formula \eqref{decomposition of LT} we also observe that $\varepsilon_{1+N_t}$ is independent of $L_t+ L^{+}_T- L^{+}_t$ given $\mathcal F_\infty^{\mathcal C}\otimes\mathcal F_T^\lambda$. In addition, by Assumption \ref{ass:1} the conditional law of $\varepsilon_{1+N_t}$ given $\mathcal F_\infty^{\mathcal C}\otimes\mathcal F_T^\lambda$ identifies with the law of $\overline{\eps}$ since $\mathcal F^\varepsilon$ is independent of $\mathcal F_T^\lambda$.\\\\
\noindent
We now compute the characteristic functions of the two random vectors of interest. Let $\chi$ be the characteristic function of the random vector \[\left(g(t,\Lambda_t,\eps_{1+N_t},\vartheta_{1+N_t}) e^{-\kappa(T-t)},L_{T}(\cdot \cup \{t\}),\lambda_t\right).\]
Let $(u_1,u_2,u_3)\in\mathbb R^3$.  One has
\[\begin{split}&\quad\;\chi(u_1,u_2,u_3):=\E\left[e^{i u_1 g(t,\Lambda_t,\eps_{1+N_t},\vartheta_{1+N_t}) e^{-\kappa(T-t)} + i u_2 L_{T}(\ndot \cup \{t\}) + iu_3 \lambda_t}\right]\\
&=\E\left[e^{iu_3 \lambda_t} e^{i u_1 g(t,\Lambda_t,\eps_{1+N_t},\vartheta_{1+N_t}) e^{-\kappa(T-t)} + i u_2 \left(L_{t} + e^{-\kappa (T-t)} (f(t,\Lambda_{t}, \eps_{1+N_t}) +  L^{+}_T - L^{+}_t\right)}\right]\\
&=\E\left[e^{iu_3 \lambda_t} e^{i u_2 (L_t+L^{+}_T-L^{+}_t)} e^{iu_2 e^{-\kappa (T-t)} f(t,\Lambda_{t},\eps_{1+N_t}) } e^{iu_1 e^{-\kappa (T-t)} g(t,\Lambda_t,\eps_{1+N_t},\vartheta_{1+N_t}) }\right].
\end{split}\]
Since $\varepsilon_{1+N_t}$ and $\vartheta_{1+N_t}$ are independent of $L_t+L^{+}_T-L^{+}_t$ given $\mathcal F_\infty^{\mathcal C}\otimes\mathcal F_T^\lambda$, we obtain that
\[\chi(u_1,u_2,u_3)=
\mathbb E\left[e^{iu_3 \lambda_t} e^{iu_2 e^{-\kappa (T-t)} f(t,\Lambda_{t},\overline{\varepsilon}) } e^{iu_1e^{-\kappa (T-t)} g(t,\Lambda_{t},\overline{\varepsilon},\overline{\vartheta}) } \mathbb E\Big[e^{i \mu_2 (L_t+ L^{+}_T-L^{+}_t)}\,\Big|\,\mathcal F^{\mathcal C}_\infty\otimes\mathcal F_T^\lambda\Big]\right],\]
where we also use the fact that the probability law of $(\varepsilon_{1+N_t},\vartheta_{1+N_t})$ given $\mathcal F_\infty^{\mathcal C}\otimes\mathcal F^\lambda_T$ coincides with $\mu$ (which, we recall, is the probability law of $(\overline{\varepsilon},\overline{\vartheta})$).  Moreover, from \eqref{third term} we observe that $L_t+{L}^{+}_T-{L}^{+}_t$ has the same law as $L_T$ conditioned on $\mathcal F_\infty^{\mathcal C}\otimes\mathcal F_T^\lambda$. Therefore, we obtain
\[\begin{split}\chi(u_1,u_2,u_3)&=\mathbb E[e^{iu_3\lambda_t}e^{iu_2 e^{-\kappa(T-t)}f(t,\Lambda_t,\overline{\varepsilon})} e^{iu_1e^{-\kappa(T-t)}g(t,\Lambda_t,\overline{\varepsilon},\overline{\vartheta})} e^{iu_2L_T}]\\
&=\mathbb E[e^{iu_2e^{-\kappa(T-t)}f(t,\Lambda_t,\overline{\varepsilon})}e^{iu_1(e^{-\kappa(T-t)}g(t,\Lambda_t,\overline{\varepsilon},\overline{\vartheta})+L_T)}e^{iu_3\lambda_t}],
\end{split}\]
which shows that $\chi$ coincides with the characteristic function of the vector
\[\left(g(t,\Lambda_t,{\overline{\varepsilon}},{\overline{\vartheta}}) e^{-\kappa (T-t)}, L_T + f(t,\Lambda_{t},{\overline{\varepsilon}}) e^{-\kappa (T-t)},\lambda_t\right).\]
The lemma is thus proved.
\end{proof}

We now turn to the statement and the proof of the main result of this paper.

\begin{theorem}\label{th:main}
Recall that $(\eps_i,\vartheta_i)_{i \in \mathbb{N}^*}$ and $(\overline \eps,\overline \vartheta)$ are i.i.d. with common law $\mu$. Under the Assumptions \ref{ass:1} and \ref{ass:2}, it holds that
\begin{align}\label{eq:pricingSL2}
& \E\left[\hat L_T h\left(L_T\right)\right] \nonumber\\
&=\int_0^T e^{-\kappa (T-t)}  \E\left[ g(t,\Lambda_t,\overline \eps, \overline \vartheta)  \, \lambda_t \, \varphi_{\lambda}^h\left(f(t,\Lambda_{t},\overline \eps) e^{-\kappa (T-t)}\right) \right] dt\nonumber\\
&=\int_{\mathbb R_+^2}\int_0^T  e^{-\kappa (T-t)} \E\left[ g(t,\Lambda_t, x,y)  \, \lambda_t \, \varphi_{\lambda}^h\left(f(t,\Lambda_{t},x) e^{-\kappa (T-t)}\right) \right] \mu(dx,dy) \, dt,
\end{align}
where $\hat L_T$ is defined in (\ref{eq:lossbis}) and  the mapping $\varphi_\lambda^h(x):=\E\left[h(L_T+x)\vert \F_T^\lambda\right]$ is defined in (\ref{eq:definitionvarphi}).
\end{theorem}

\begin{proof}
Assumptions \ref{ass:1} and \ref{ass:2} are in force.
Using the {relation \eqref{eq:Lbis}} and the integration by parts formula on the Poisson space \eqref{eq:IPP}, it holds that
\[\begin{split}&\quad\;\E\left[\hat L_T h\left(L_T\right) \right]=\E\left[\E\left[\hat L_T h\left(L_T\right)\middle\vert \F^{\eps,\vartheta} \vee \F_T^\lambda \right]\right]\\
&=\E\left[\E\left[h\left(L_T\right) \int_0^T Z_t dN_t\middle\vert \F^{\eps,\vartheta} \vee \F_T^\lambda \right]\right]=\E\left[\int_0^T Z_t h\left(L_T(\cdot \cup {\{t\}})\right) \lambda_t dt \right]
\end{split}\]
By Relation \eqref{equ:Zs2} and the fact that the set $\{\Delta_t N\neq 0\}$ is negligeable, we obtain
\[\begin{split}\E\left[\hat L_T h\left(L_T\right) \right]&=\E\left[\int_0^T g(t,\Lambda_t,\eps_{1+N_t},\vartheta_{1+N_t}) e^{-\kappa (T-t)}  h\left(L_T(\cdot \cup \{t\})\right) \lambda_t dt \right]\\
&=\int_0^T \E\left[g(\Lambda_t,\eps_{1+N_t},\vartheta_{1+N_t}) e^{-\kappa (T-t)}  h\left(L_T(\cdot \cup \{t\})\right) \lambda_t\right] dt .\end{split}\]
{Finally, by Lemma \ref{lemme:temp1lemma}, the above formula leads to }
\[\E\left[L_T \ind{L_T \in [K,M]} \right]=\int_0^T\E\left[ g(t,\Lambda_t,\overline{\varepsilon},\overline{\vartheta}) e^{-\kappa (T-t)} h\left(L_T + f(t,\Lambda_{t},\overline{\varepsilon}) e^{-\kappa (T-t)}\right) \lambda_t \right] dt.\]
Since $\overline{\varepsilon}$ is independent of $\mathcal F^\lambda_T\vee\mathcal F^{\varepsilon}$, one has
\[\mathbb E\left[h\left(L_T + f(\Lambda_{t},\overline{\varepsilon}) e^{-\kappa (T-t)} \right)\,\middle\vert\,\mathcal F_T^\lambda\vee\sigma(\overline{\varepsilon})\right]=\varphi_\lambda^h\left(f(t,\Lambda_t,\overline{\varepsilon})e^{-\kappa(T-t)}\right).\]
Therefore
\[\begin{split}&\quad\;\E\left[\hat L_T h\left(L_T\right) \right]=\int_0^T\E\left[ g(t,\Lambda_t,\overline \eps, \overline \varepsilon) e^{-\kappa (T-t)} \lambda_t \varphi_{\lambda}^h\left(f(t,\Lambda_{t},\overline \eps) e^{-\kappa (T-t)},f(t,\Lambda_{t},\overline \eps) e^{-\kappa (T-t)}\right) \right] dt\\
&=\int_{\mathbb R_+^2}\int_0^T  e^{-\kappa (T-t)} \E\left[ g(t,\Lambda_t, x,y)  \, \lambda_t \, \varphi_{\lambda}^h\left(f(t,\Lambda_{t},x) e^{-\kappa (T-t)}\right) \right] dt\,\mu(dx,dy),
\end{split}\]
as asserted by the theorem.
\end{proof}

\begin{remark}\begin{enumerate}
\item
Note that from Equality (\ref{eq:pricingSL2}), it is clear that our approach only requires the knowledge of the conditional law of $L_T$ given $\lambda$ (\textit{via} the mapping $\varphi_\lambda$) and not the one of the pair $(L_T,\hat L_T)$. 
This seems to be particularly useful for the numerical approximation of the aforementioned expectation.
\item The theorem above provides us the relation of the pricing formula with respect to the intensity process $(\lambda_t)_{t\geq 0}$ of the counting process.
\end{enumerate}
\end{remark}

Relation (\ref{eq:pricingSL2}) allows us to give a lower (respectively upper) bound on the price if $h$ is assumed to be convex (respectively concave).
\begin{corollary}
Under the assumptions of Theorem \ref{th:main}, it holds that :
\begin{itemize}
\item[(i)] if $h$ is convex, then
\begin{align*}
&\E\left[\hat L_T h\left(L_T\right)\right] \\
&\geq \int_{\mathbb R_+^2}\int_0^T  e^{-\kappa (T-t)} \E\left[ g(t,\Lambda_t, x,y)  \, \lambda_t \, h\left(\E[L_T\middle\vert \mathcal F_T^\lambda]+f(t,\Lambda_{t},x) e^{-\kappa (T-t)}\right) \right] \mu(dx,dy) \, dt.
\end{align*}
\item[(i)] if $h$ is concave, then
\begin{align*}
&\E\left[\hat L_T h\left(L_T\right)\right] \\
&\leq \int_{\mathbb R_+^2}\int_0^T  e^{-\kappa (T-t)} \E\left[ g(t,\Lambda_t, x,y)  \, \lambda_t \, h\left(\E[L_T\middle\vert \mathcal F_T^\lambda]+f(t,\Lambda_{t},x) e^{-\kappa (T-t)}\right) \right] \mu(dx,dy) \, dt.
\end{align*}
\end{itemize}
\end{corollary}

\begin{proof}
We prove (i) as statement (ii) follows the same line. As $h$ is assumed to be convex, Jensen's inequality implies that
$$ \varphi_\lambda^h(x) \geq h\left(\E[L_T\middle\vert \mathcal F_T^\lambda]+x\right), \quad x \in \real_+. $$
The result is then obtained by plugging this estimate in Relation (\ref{eq:pricingSL2}).
\end{proof}

\section{Applications and examples} 
\label{section:examples}

In this section, we provide some application examples of our main result, in particular for the (generalized) stop-loss contract. Such explicit computations will also be useful for the CDO tranches and expected shortfall risk measure.

\subsection{Computation of the building block}

We first focus on the building block $\varphi_\lambda^{h}$ (defined in (\ref{eq:definitionvarphi})) when $h:=\ind{[K,M]}$ :
$$\varphi_{\lambda}(x) := \varphi_\lambda^{h}(x) = \P\left[L_T \in [K-x,M-x]\vert \F_T^\lambda\right], \quad x\in \real_+$$
which corresponds to the payoff of a stop-loss contract or a CDO tranche.
Let $\F^{\eps}:=\sigma(\eps_i,\; i\in\N^*)$. For any $i\in\N^*$, we set $X_i:=f(\tau_i,\Lambda_{\tau_i},\eps_i)$ and $x$ in $\real_+$, we have
\begin{align}
\label{eq:briqueexemple}
&\P\left[L_T \in [K-x,M-x] \middle\vert \F^\lambda_T \vee \F^\eps\right]\nonumber\\
&=\P\left[\sum_{i=1}^{N_T} X_i e^{\kappa \tau_i} \in [(K-x) e^{\kappa T},(M-x) e^{\kappa T}] \Bigg\vert \F^\lambda_T \vee \F^\eps\right]\nonumber\\
&=\sum_{k=1}^{+\infty} \E\left[ \sum_{i=1}^{k} X_i e^{\kappa \tau_i} \in [(K-x) e^{\kappa T},(M-x) e^{\kappa T}] \Bigg\vert N_T=k,\F^\lambda_T \vee \F^\eps \right] \P[N_T=k\vert \F_T^\lambda]\nonumber\\
&= \sum_{k=1}^{+\infty} e^{-\int_0^T \lambda_s ds } \int_{\mathcal S_k} \P \left[\sum_{i=1}^{k} X_i e^{\kappa t_i} \in [(K-x) e^{\kappa T} ,(M-x) e^{\kappa T} ]  \Bigg\vert \F^\lambda_T \vee \F^\eps\right] \lambda_{t_1} dt_1\cdots \lambda_{t_k} dt_k\nonumber\\
&= \sum_{k=1}^{+\infty} e^{-\int_0^T \lambda_s ds } \int_{\mathcal S_k} \int_{\real_+^k} \ind{\sum_{i=1}^{k} x_i e^{\kappa t_i} \in [(K-x) e^{\kappa T} ,(M-x) e^{\kappa T} ]}  \mathcal{L}_{X_{(1:k)}}^{\vert \lambda}(dx_1,\ldots,dx_k) \lambda_{t_1} dt_1\cdots \lambda_{t_k} dt_k,
\end{align}
where $\mathcal S_k:=\{0<t_1<\cdots < t_k \leq T\}$, $X_{(1:k)}:=(X_1,\ldots,X_k)$ and
$$ \mathcal{L}_{X_{(1:k)}}^{\vert \lambda}(dx_1,\ldots,dx_k):=\P\left[X_{(1:k)} \in (dx_1,\ldots,dx_k)\Bigg\vert \F^\lambda_T\right].$$
It just remains to compute the joint distribution of the claims $X_{(1:k)}$ in different situations. In particular, we provide below an explicit example.\\\\
\noindent
\textbf{Model on $\varepsilon_i$ : } We assume that $(\eps_i)_{i\in\N^*}$ are i.i.d. random variables with Pareto distributions $\mathcal{P}(\alpha_{\varepsilon},\beta_{\varepsilon})$ with $(\alpha_{\varepsilon},\beta_{\varepsilon}) \in (\real^*_+)^2$ whose density $\psi_\eps$ is defined as
$$ \psi_\eps(z)=\left(\beta_\eps\frac{\alpha_\eps^{\beta_\eps}}{z^{\beta_\eps+1}}\right){\bf 1}_{\{z \geq \alpha_\eps\}}dz.$$
Choosing $f(t,\ell,x):= \sqrt{\frac{\ell}{t}}x$, the conditional distribution $\mathcal{L}_{X_{(1:k)}}^{\vert \lambda}(dx_1,\ldots,dx_k)$ in Relation (\ref{eq:briqueexemple}) becomes
\begin{align*}
&\mathcal{L}_{X_{(1:k)}}^{\vert \lambda}(dx_1,\ldots,dx_k)\\
&=\P\left[\left(\sqrt{\frac{\Lambda_{t_1}}{t_1}} \eps_1,\cdots,\sqrt{\frac{\Lambda_{t_k}}{t_k}} \eps_k\right) \in (dx_1,\ldots,dx_k)\Bigg\vert \F^\lambda_T\right]\\
&=\prod_{i=1}^k \P\left[\sqrt{\frac{\Lambda_{t_i}}{t_i}} \eps_i \in dx_i\Bigg\vert \F^\lambda_T\right]\\
&=\prod_{i=1}^k \left(\beta_\eps\frac{\left(\sqrt{\frac{t_i}{\Lambda_{t_i}}}\alpha_\eps\right)^{\beta_\eps}}{z_i^{\beta_\eps+1}}\right){\bf 1}_{\left\{z_i \geq \sqrt{\frac{t_i}{\Lambda_{t_i}}} \alpha_\eps\right\}}dz_i.
\end{align*}
The next step to compute the right-hand side of Relation (\ref{eq:pricingSL2}) is to specify the joint law of $(\eps_1,\vartheta_1)$.\\\\
\noindent
\textbf{Model on $(\varepsilon_i,\vartheta_i)$ : } We assume that $(\varepsilon_i,\vartheta_i)_{i\in\N^*}$ are i.i.d. random vectors, with marginal distributions following Pareto distributions $\mathcal{P}(\alpha_{\varepsilon},\beta_{\varepsilon})$ and $\mathcal{P}(\alpha_{\vartheta},\beta_{\vartheta})$ (for a set of parameters ${\alpha_\eps, \, \beta_\eps, \, \alpha_\vartheta, \, \beta_\vartheta > 0}$) respectively. The dependence structure is modeled through a Clayton copula with parameter $\theta>0$.
We recall that the Clayton copula is $C(u,v):=\left(u^{-\theta}+v^{-\theta}-1\right)^{-\frac{1}{\theta}}$ and the density $c$ of the Clayton copula is given by $$c(u,v):= (1+\theta) (uv)^{-1-\theta}(u^{-\theta}+v^{-\theta}-1)^{-\frac{1}{\theta}-2}.$$
\noindent
The joint distribution of $\left(\varepsilon_1, \vartheta_1\right)$ is then given by
$$\mu(dx, dy)= c\left(F_{\varepsilon}(x), F_{\vartheta}(y)\right) \psi_{\varepsilon}(x) \psi_{\vartheta}(y) dx dy,$$
with $F_\eps(z)=\left(1- \frac{\alpha_\eps^{\beta_\eps}}{z^{\beta_\eps}}\right)$ and $F_\vartheta(z)=\left(1- \frac{\alpha_\vartheta^{\beta_\vartheta}}{z^{\beta_\vartheta}}\right)$.\\\\
\noindent
\textbf{Joint law of $(\lambda_t,\Lambda_t)$ : } The final step in the computation of Relation (\ref{eq:pricingSL2}) is to make precise the joint law of $(\lambda_t,\Lambda_t)$. More precisely, we need to compute $$\E   \left[ g(t,\Lambda_t,x, y) \, \lambda_t \, \varphi_{\lambda}\left(K-f(t,\Lambda_{t},x) e^{-\kappa (T-t)},M-f(t,\Lambda_{t},x) e^{-\kappa (T-t)}\right) \right]$$
\noindent
Assume  the intensity process $ (\lambda_t)_{t\in[0, T]}$ is given by
$$\lambda_t= \lambda_0 \exp(2 \beta W_t)$$   where $W$ is a Brownian motion, and $\beta$ a constant (non null).
Then the cumulative intensity  is $$\Lambda_t=\lambda_0 \int_0^t \exp(2 \beta W_s)ds .$$
By  Borodin and Salminen \cite{borodin2012handbook} (page 169), the joint law of $(\Lambda_t, W_t)$ is given by
$$ \mathbb P \left( \Lambda_t \in dv ,   W_t \in dz          \right) = \frac{\lambda_0 |\beta|}{2v} \exp\left(  -  \frac{\lambda_0 (1+ e^{2\beta z})}{2 \beta^2 v}\right) i_{\beta^2t/2}\left(    \frac{\lambda_0 e^{\beta z}}{\beta^2v} \right) \ind{v>0}   dv dz$$
where the function
$$i_y(z)=\frac{z e^{\frac{\pi ^2}{4y}}}{\pi \sqrt{\pi y}} \int_0^{\infty} \exp\left(-z ch(x)  -\frac{x^2}{4y} \right) sh(x) \sin(\frac{\pi x}{2y}) dx.$$
The expectation term in the right-hand side of equation \eqref{eq:pricingSL2} is then

$$\E   \left[ g(t,\Lambda_t,x, y) \, \lambda_t \, \varphi_{\lambda}\left(K-f(t,\Lambda_{t},x) e^{-\kappa (T-t)},M-f(t,\Lambda_{t},x) e^{-\kappa (T-t)}\right) \right]$$
{\tiny $$= \int_{\mathbb R^2}  g(t,v,x, y)  {\tiny \varphi_{\lambda}
\left(K-\sqrt{\frac{v}{t}}x e^{-\kappa (T-t)},M-\sqrt{\frac{v}{t}}xe^{-\kappa (T-t)}\right) }
\frac{\lambda_0^2 |\beta|}{2v}  e^{2 \beta z} \exp\left(  -  \frac{\lambda_0 (1+ e^{2\beta z})}{2 \beta^2 v}\right) i_{\beta^2t/2}\left(    \frac{\lambda_0 e^{\beta z}}{\beta^2v} \right) \ind{v>0}    dv dz.$$}

\subsection{A Black-Scholes type formula for generalized Stop-Loss contracts in the Cramer-Lundberg}

As an illustration, we conclude our analysis by specifying our result in the classic Cramer-Lundberg model. More precisely, we assume that the Cox process is an homogeneous Poisson process with constant intensity $\lambda_0>0$ and set $h:=\textbf{1}_{[K,M]}$, with $K<M$. The building block reduces to the distribution function
\begin{equation}
\label{eq:buildblocklambdaconst}
\varphi_{\lambda_0}(x):=\varphi_{\lambda_0}^h(x)= \P\left[L_T \in [K-x,M-x] \right], \quad x\in \real_+.
\end{equation}
In that case we omit the dependency on $\Lambda$ for the mappings $f$ and $g$ (as $\Lambda_t=t \lambda_0)$.
\begin{corollary}
Under the assumptions of Theorem \ref{th:main}, it holds
\begin{align*}
& \E\left[\hat L_T \textbf{1}_{L_T\in [K,M]}\right]=\lambda_0 \int_0^T \int_{\real_+^2} e^{-\kappa (T-t)}  g(t,x,y) \, \varphi_{\lambda_0}\left(f(t,x) e^{-\kappa (T-t)}\right) \mu(dx,dy)dt,
\end{align*}
(recall that $\mu:=\mathcal{L}_{(\overline \eps,\overline \vartheta)}$).
\end{corollary}



If we assume furthermore that $f(t,x)=g(t,x,y)=x$ and $\kappa=0$, then the loss process $L_T$ corresponds to the cumulated loss of the classic Cramer-Lundberg model. In this context, a huge literature deals with the computation of the ruin probability and related quantities, such as the discounted penalty function at ruin (Gerber-Shiu function). Others papers are concerned with the pricing of Stop-Loss contract. The pricing relies on the computation of a term of the form $\int_K^{M} y dF(y)$ with $F$ being the cumulative distribution function of the loss process $L_T$, and the discussion in the literature mainly focuses on  the derivation of the compound distribution function $F$ (usually calculated recursively, using Panjer recursion formula and numerical methods/approximations) \textit{cf.} \cite{panjer1981} and \cite{gerber1982}.
Our Malliavin {approach} provides another formula which reads as
\begin{equation}
\label{eq:OFCL}
\E\left[\hat L_T \textbf{1}_{L_T\in [K,M]}\right]=\lambda_0 T \int_{\real_+} x \, (F(M-x)-F(K-x)) \mu(dx).
\end{equation}
{Note that our result indeed coincides with the one obtained in \cite{gerber1982}. Indeed, if one translates in a general setting Formula \cite[(6)]{gerber1982} (as $\mu$ is constrained to have a finite support in $\mathbb{N}$ in \cite{gerber1982}) the distribution $F$ satisfies  
$$y dF(y)= \lambda_0 T \int_{\real_+} x  dF(y-x) \mu(dx),$$
from which one deduces that 
\begin{align*}
\int_K^M y dF(y) &= \lambda_0 T \int_K^M \int_{\real_+} x  \mu(dx) dF(y-x) \\
&= \lambda_0 T \int_{\real_+} x \int_K^M dF(y-x) \mu(dx)\\
&= \lambda_0 T \int_{\real_+} x (F(K-x)-F(M-x)) \mu(dx),
\end{align*}
which is exactly (\ref{eq:OFCL}).
}

\bibliographystyle{plain}
\bibliography{biblioHJR}

\end{document}